\documentclass[12pt]{iopart}

\usepackage[english]{babel}
\usepackage[utf8]{inputenc}
\usepackage[colorinlistoftodos, color=green!40, prependcaption]{todonotes}
\usepackage{color}
\usepackage{titlesec}
\usepackage{amsthm}
\usepackage{mathtools}
\usepackage{amsmath}
\usepackage{mathrsfs}
\usepackage{physics} 
\usepackage{xcolor}
\usepackage{graphicx}
\usepackage{soul}
\usepackage{bm}
\usepackage{comment}
\usepackage{xcolor}
\usepackage{MnSymbol}%
\usepackage{wasysym}
\usepackage{cite}
\usepackage{tabularx,ragged2e}
\usepackage{tikz-cd}
\usepackage{amscd}
\usepackage{mathrsfs}  
\usepackage{float}
\usepackage{tikz,listofitems}
\usepackage{enumitem}

\usepackage{mathbbol}
\usepackage[pdftex, pdftitle={Article}, pdfauthor={Author}]{hyperref} 

\newtheorem{theorem}{Theorem}
\theoremstyle{definition}
\newtheorem{remark}{Remark}
\newtheorem{definition}{Definition}
\newtheorem{example}{Example}

\newtheorem{lemma}{Lemma}
\newtheorem{proposition}{Proposition}

\newtheorem{problem}{Problem}
\newtheorem{weak Nullstellensatz}{weak Nullstellensatz}
\DeclarePairedDelimiter\floor{\lfloor}{\rfloor}

\begin{document}

\article[On TI MPS and advances in MPS representations of the $W$-state]{}{On Translation-Invariant Matrix Product States and advances in MPS representations of the $W$-state}

\author{Petr Klimov,  Richik Sengupta$^1$, and Jacob Biamonte\footnote{Former address}$^1$}
\address{$^1$Skolkovo Institute of Science and Technology, Moscow, Russian Federation}

\ead{peterklimov@yandex.ru, r.sengupta@skoltech.ru, jacob.biamonte@deepquantum.ai}
    
\begin{abstract}
This work is devoted to the study Translation-Invariant (TI) Matrix Product State (MPS) representations of quantum states with periodic boundary conditions (PBC). We pursue two directions: we introduce new methods for constructing TI MPS representations for a certain class of TI states and study their optimality in terms of their bond dimension. We pay particular attention to the  $n$-party $W$-state and construct a TI MPS representation of bond dimension $\floor*{\frac{n}{2}}+1$ for it. We generalize the approach implemented for the $W$-state to obtain TI MPS representations of a larger class of states satisfying several structural conditions. We further study properties of this class and show that we can can always achieve a bond dimension of $n$ for TI MPS representation of states in this class. In the framework of studying optimality of TI MPS representations with PBC, we study the optimal bond dimension $d(\psi)$ for a given state $\psi$. In particular we introduce a deterministic algorithm for the search of $d(\psi)$ for an arbitary state. Using numerical methods, we verify the optimality of our previous construction for the $n$-party $W$-state for small $n$.

\end{abstract}

\section{Introduction}

The concept of MPS (Matrix Product State) is an important concept in quantum theory, especially in quantum information theory.  In recent years, MPS representations have become an essential tool for investigating condensed matter physics, quantum information theory, and quantum field theory. It arose naturally in tensor networks, and is one of the most studied of tensor networks \cite{Gar+06, Fannes92, Klumper92, Silvi2012TensorNA}. It is also sometimes referred to as tensor trains \cite{Ose11}. Moreover, similar concepts are actively used in other areas like classical and quantum machine learning \cite{10.1162/neco_a_01202, Meng2020, Dborin_2022, NIPS2016_5314b967, sengupta2022tensor}. 

In general, for a quantum state $\psi$ (or simply for a vector $\psi \in (\mathbb{C}^2)^{\otimes n}$), the MPS representation is written as
$$
|\psi\rangle=\sum_{i_1,\ldots,i_N=1}^2
\trace\left[A_{i_1}^{[1]}A_{i_2}^{[2]}\cdots
A_{i_N}^{[N]}\right]|i_1,i_2,\ldots,i_N\rangle\;,
$$ where $A_i^{[k]}$ are complex matrices.
This representation is also sometimes called MPS with PBC (periodic boundary conditions) to distinguish it from the analogous  MPS representation with OBC (open boundary conditions).

A convenient and interesting class of MPS representations to consider is TI (translationally invariant) MPS with periodic boundary conditions (PBC)\cite{Gar+06}. This representation can be used when the matrices are site-independent i.e $A_1^{[k]} = A_1$ for all $k$. This type of MPS representation uses significantly lower number of matrices in the representation, even though sometimes at a cost of increase in the dimension of the matrices in the representation often referred to as bond dimension.  It has been proved that any TI state $\psi$ admits a TI MPS representation with PBC and moreover an upper bound on $d(\psi)$  (minimal possible dimension of the matrices in the  TI MPS  representation with PBC for $\psi$) can be obtained that depends on the dimension of the state $\psi$ \cite{Gar+06}.

However, a big unsolved problem is obtaining exact estimates or improving existing estimates of the dimension $d(\psi)$ at least for certain classes of TI states, for example the $W$-state.
In \cite{Gar+06}, a TI MPS representation with PBC  was provided for the $W$-state of order $n$ with bond dimensions of the order $O(n)$ with constant factor $1$. Moreover, in \cite{Gar+06} it was hypothesized that the bond dimension for the MPS representation of the $W$-state is lower bounded by $O(n^{\frac{1}{3}})$. In \cite{ShitovMPS}, a weaker form of this hypothesis was proved which stated that for all $\delta > 0$  the bond dimension of the $W$-state is lower bounded by $O(n^{\frac{1}{3 + \delta}})$.

In \textbf{Theorem \ref{W state representation}} of the current paper we construct a TI MPS representation with PBC of the  $W$-state of dimension $n$ with matrices of size $\floor*{\frac{n}{2}}+1$. This improves the previous result known to the authors on the TI MPS representation with PBC  for the $W$-state with minimum bond dimension.

In \textbf{Theorem \ref{base i-sparse TI MPS construction}} we extend the method previously developed by us in the previous theorem to build TI MPS representation with PBC for a large class of "sparse" states. This generalization provides a framework to search for more optimal TI MPS representations with PBC.

\textbf{Theorem \ref{canonical construction with matrix unit}} we demonstrate that all TI MPS representations with PBC, constructed using the approach outlined in Theorem \ref{base i-sparse TI MPS construction}, can be transformed into a lower-dimensional representation with matrices of size $n\times n$. Furthermore, this new reduced representation can be obtained directly from the original representation using the formulae specified in the theorem.

As previously mentioned  any TI state $\psi$ admits a TI MPS representation with PBC, taking this fact into account we explicitly construct a search algorithm for  $d(\psi)$ for an arbitrary state $\psi$ in \textbf{Theorem \ref{bb}}. The algorithm allows us to deterministically obtain $d(\psi)$ for a TI MPS with PBC representation of a given state. 

The results of our numerical experiments for obtaining  $d(\psi)$ for the $W$-state using the algorithm from Theorem \ref{bb} hints towards the possibility  that the estimate in Theorem \ref{W state representation} could be optimal.

\section{Notation}

Let us formulate the main definitions and objects with which we will work.
The standard basis in the vector space $\mathbb{C}^2$ comprises of vectors $(1, 0)$ and $(0, 1)$. Quantum computation uses the Dirac notation, where they are denoted by $\ket{0}$ and $\ket{1}$, respectively, and they are said to form the \textit{computational basis}.

We consider the $n$ tensor power of the space $\mathbb{C}^2$ denoted by $(\mathbb{C}^2)^{\otimes n}$, whose basis comprises all possible $n$ tensor products of the standard basis vectors of $\mathbb{C}^2$ of the form $\ket{i_1} \otimes \ket{i_2} \ldots \ket{i_n}$, where $(i_1,\ldots,i_n) \in \{0, 1\}^n$. In Dirac notation these products are denoted as $\ket{i_1\ldots i_n}$.  In other words, formal strings of length $n$ comprising of numbers $0$ and $1$ encode the corresponding basis vectors. In total, there are $2^n$ basis vectors in $(\mathbb{C}^2)^{\otimes n}$.

We call an arbitrary vector from $(\mathbb{C}^2)^{\otimes n}$ with a unit norm a \textit{quantum state}. We can interpret vectors with non-unitary norm as non-normalized quantum states. By default, theorems in this paper are suitable both for normalized and non-normalized quantum states. In the case when we intend to underline that the result is specifically for non-normalized or normalized state, we will do so.

Every $ \psi \in (\mathbb{C}^2)^{\otimes n} $ can
  be decomposed in the computational basis with complex coefficients as:
  \begin{align}\label{StateDecomposition}
  \ket{\psi} = \sum\limits_{(i_1,\ldots,i_n) \in \{0, 1\}^n} c_{(i_1,\ldots,i_n)}\ket{i_1i_2\ldots i_n}
  \end{align}
  
  or in a more compressed form
  \begin{align}\label{StateShortDecomposition}
  \ket{\psi} = \sum\limits_{I \in \{0, 1\}^n} c_I\ket{I}.
  \end{align}
 
 Thus, we denote the coefficients of the vector $ \psi \in (\mathbb{C}^2)^{\otimes n} $ as $c_I^{\psi}$ or simply $c_I$, if the state is clear from the context.
 
We call the quantum state $ \psi \in (\mathbb{C}^2)^{\otimes n} $   \textit{translationally invariant} or \textit{TI state}, if the coefficients do not change under cyclic shifts of the basis vectors, i.e.
\begin{align}\label{TIState}
c_{(i_1, i_2, ..., i_n)} = c_{(i_2, i_3, \ldots, i_n, i_1)}\ \forall (i_1,\ldots,i_n) \in \{0, 1\}^n.
\end{align}

Let us denote the set of all TI states from $(\mathbb{C}^2)^{\otimes n}$ as $TI-(\mathbb{C}^2)^{\otimes n}$. 
 
For fixed $\psi \in (\mathbb{C}^2)^{\otimes n}$ its coefficients can be considered as functions $c: \{0, 1\}^n \to \mathbb{C}$ or as a tensor. The latter makes it possible to consider Matrix Product States (MPS) representations of the given state. There are various forms of MPS representations, we will work in this paper with TI MPS representations with Periodic Boundary Conditions (PBC)  of TI states.

Following is a general mathematical formulation of MPS representation with PBC for any $\psi \in (\mathbb{C}^2)^{\otimes n}$. 

\begin{definition}
 MPS representation with PBC for $\psi \in (\mathbb{C}^2)^{\otimes n}$ has the form
 \begin{align}\label{StateShortDecomposition}
 |\psi\rangle=\sum_{i_1,\ldots,i_N=0}^1
 \Trace\left[A_{i_1}^{[1]}A_{i_2}^{[2]}\cdots
 A_{i_N}^{[N]}\right]|i_1,i_2,\ldots,i_N\rangle\;,
 \end{align}
 where $A_{i_k}^{[k]}$ are complex matrices of dimension $d_k \times d_{k + 1}$. 
\end{definition}

When the matrices $A_{i_k}^{[k]} = A_{i_k} $ for all $ k= \overline{1,N}$, they are said to be "site-independent", we call such an MPS representation with PBC, a TI (translationally invariant) MPS representation with PBC. In this case, the dimensions of the matrices $A_{i_k}$ coincide, that is, $d_{i_k} = d$. The dimension $d$ is also called the \textit{bond dimension}. Given that an MPS representation with PBC is uniquely determined by the matrices used in the representation, we can reformulate the definition of TI MPS representation with PBC in the following convenient way.

\begin{definition}
TI MPS representation with PBC of bond dimension $d$ for the TI state $\psi \in (\mathbb{C}^2)^{\otimes n}$ is a pair of $d\times d$ complex matrices $A_0$ and $A_1$, such that
\begin{align}\label{base TI formulation}
\hspace*{-1cm} \Tr (A_{i_1} A_{i_2} \ldots A_{i_n}) = c_{(i_1, \ldots, i_n)} \ \ \ \forall (i_1, \ldots, i_n) \in \{0, 1\}^n.
\end{align}
\end{definition}

Note that by TI MPS with PBC we mean site-independent MPS representations for TI states with periodic boundary conditions. Some authors mention it explicitly, but conventionally it is omitted. 

Due to the cyclicity of the trace, finding the maximal number of unique equations in \ref{base TI formulation} is equivalent to solving the famous necklace problem from combinatorics. Using Polya's enumeration theorem\cite{RedfieldPolya_theorem, Polya_theorem}, the maximal number of unique equations is:

\begin{align}\label{numb}
    \frac {1}{n}\sum _{p|n}\varphi (p)2^{n/p},
\end{align}

where $\phi(p)$ is the Euler totient function.

TI MPS representations with PBC allow us to store information about the coefficients in two matrices and obtain them by taking the traces of the products of these matrices. There are also other forms of expressing coefficients in terms of MPS.

\begin{proposition}\label{lamb}
    Let  $A_0$ and $A_1$ determine the TI MPS representation with PBC for TI state $\psi \in (\mathbb{C}^2)^{\otimes n}$. Then $A_0' = \sqrt[n]{\lambda}A_0$ and $A_1'=\sqrt[n]{\lambda}A_1$ determine TI MPS representation with PBC for $\lambda \psi$ where $\sqrt[n]{\lambda} \in \mathbb{C}$ can be any of the $n$-th roots of $\lambda$.
\end{proposition}
\begin{proof}
From linearity of trace we have

$\Tr (A_{i_1}' A_{i_2}' \ldots A_{i_n}')=\Tr (\sqrt[n]{\lambda}A_{i_1} \sqrt[n]{\lambda}A_{i_2} \ldots \sqrt[n]{\lambda}A_{i_n})=\lambda\Tr (A_{i_1} A_{i_2} \ldots A_{i_n}).$

\end{proof}

This proposition in particular allows one to search for TI MPS representations with PBC for quantum states in non-normalized form and obtain the representations of quantum states in the normalized form essentially for free. Specifically, the matrices for the normalized states can be obtained by multiplying a suitable constant factor to matrices obtained for the non-normalized states.

Some quantum states are important in quantum theory, one such state is the $W$-state\cite{Zhu_2015}. 

\begin{definition}
The $W$-state of order $n$ is defined as 
\begin{align}\ket{W_n} = \frac{1}{\sqrt n}\sum\limits_{(i_1,\ldots,i_n) \in \{0, 1\}^n : \sum\limits_{j=1}^n i_j = 1} \ket{i_1\ldots i_n}
\end{align}
\end{definition}

It is easy to see that the $W$-state is a TI state. The problem of finding a TI MPS representation with PBC for the $W$-state follows from the following definition:

\begin{definition}
TI MPS representation with PBC of bond dimension $d$ for the $W$-state are pair of complex matrices $A_0$ and $A_1$ such that
\begin{align}\label{base formulation normalized}
\Tr (A_{i_1} A_{i_2} \ldots A_{i_n}) =  \begin{cases*}
 0\ \ \  \forall (i_1, \ldots, i_n) \in \{0, 1\}^n \text{ such that } \sum\limits_{j=1}^n i_j \ne 1\\
 \dfrac{1}{\sqrt n} \ \ \  \forall (i_1, \ldots, i_n) \in \{0, 1\}^n \text{ such that } \sum\limits_{j=1}^n i_j = 1
\end{cases*}
\end{align}
\end{definition}

Proposition \ref{lamb} allows us to search for TI MPS representations with PBC for a state multiplied by some constant instead of the original. 

\begin{remark}\label{1}
We can find TI MPS representation with PBC of bond dimension $d$ for non-normalized (multiplied by $\sqrt{n}c$ where $c  \in \mathbb{C} / \{0\}$) $W$-state of order $n$ using complex matrices $A_0$ and $A_1$ such that
\begin{align}\label{base formulation}
\Tr (A_{i_1} A_{i_2} \ldots A_{i_n}) =  \begin{cases*}
 0\ \ \  \forall (i_1, \ldots, i_n) \in \{0, 1\}^n \text{ such that } \sum\limits_{j=1}^n i_j \ne 1\\
 c \ \ \  \forall (i_1, \ldots, i_n) \in \{0, 1\}^n \text{ such that } \sum\limits_{j=1}^n i_j = 1
\end{cases*}
\end{align}
From our previous discussion, it is easy to see that we can obtain a TI MPS representation with PBC for the (normalized) $W$-state using $A_0' = \frac{1}{\sqrt[n]{c\sqrt{n}}} A_0$, $A_1' = \frac{1}{\sqrt[n]{c\sqrt{n}}} A_1$ with $c = Tr(A_0^{n - 1} A_1)$.
\end{remark}

\begin{definition}
We call $d(\psi)$ the minimal bond dimension $d$ such that there exists TI MPS representation with PBC for the TI state $\psi$.
\end{definition}

From \cite{Gar+06} it follows that $d(\psi)$ is well defined (since TI MPS representation with PBC exists for every TI state). If we have some state that is defined for each $n$ (such as the $W$- state) then we can consider the function $d(\psi(n))$ which for every fixed $n$ is equal to the minimum dimension $d$, such that for corresponding state of order $n$ there exists a TI MPS representation with PBC  of bond dimension $d$.  
In the general, the problem of determining $d(\psi)$  for specific states and its asymptotics, as well as the construction of  matrices  $A_0$ and $A_1$ on which the best scaling is achieved is a difficult problem.

In paper \cite{Gar+06}, it was shown that for the $W$-state $d(n) = O(n)$ with constant factor $1$. In \cite{ShitovMPS} it was proved that $\forall \delta > 0, \ d(n) = \Omega(n^{\frac{1}{3 + \delta}})$. At the same time, the question regarding the exact asymptotics of $d(\psi(n))$ and the constant factor remains open.






\section{Approaches to building a TI MPS}

In \cite{GV03}  it was shown that for any  quantum state $\psi \in (\mathbb{C}^2)^{\otimes n}$ we can construct an MPS representation with PBC. Moreover, in \cite{Gar+06} it was shown that for any state it is possible to construct an MPS with Open Boundary Conditions (OBC). Using Theorem 3 from \cite{Gar+06} which connects TI MPS with PBC to MPS with OBC, it follows that we can construct TI MPS representation with PBC for any TI state.

However, using general techniques, we find that the MPS representation scales rapidly with the bond dimension of the matrices in the representation. For example, for MPS with OBC the dimension of the constructed matrices $d(\psi) = O(2^{\frac{n}{2}})$, and for TI MPS with PBC $d(\psi) = O(n 2^{\frac{n}{2}})$.

As far as is known to the authors of the article, there is no previously explored general way to find an exact estimate for the asymptotics of $d(\psi)$ for a given state $\psi \in (\mathbb{C}^2)^{\otimes n}$ or a method for constructing an MPS representation with the dimensions of the matrices that will not grow too fast and whose dimensions will be closer to the theoretical $d(\psi)$. Moreover, the way to build TI MPS with PBC using MPS with OBC from \cite{Gar+06} has the dimensions of the matrices in the representation at best $O(n)$ with constant factor $1$. In principle, we cannot construct representations with better asymptotics or constant factor this way.

If we want to find better TI MPS representation with PBC in terms of the bond dimension for a particular state, we need to look for a solution to the system of equations \ref{base TI formulation}. In this system, there is an exponential number of equations that depend in a non-trivial way on taking the trace of various combinations of products of matrices, which in the general case is quite difficult to solve. Therefore, methods for finding optimal TI MPS with PBC are very important, and it is useful to have ways of constructing MPS at least for some classes of states.

Returning to the $W$-state, the predominant part of coefficients in the basis decomposition is 0, which motivates to use for example a nilpotent matrix as one of the matrices so that a large number of coefficients computed as $\Tr (A_{i_1} A_{i_2} \ldots A_{i_n}) $ vanish. As we will show below we can pick the matrix unit $E_{i,j}$ (a matrix whose $(i,j)-$th entry is one and the rest of the entries are zero)  as one of the matrices and still obtain a considerably good solution.

Based on the Remark \ref{1}, it is easy to see that constructing a TI MPS representation with PBC  for some non-normalized state allows one to construct a TI MPS representation with PBC of the same bond dimension with the original matrices, normalized by some constant. 

Below we construct a TI MPS representation with PBC for the $W$-state using the idea of considering a matrix unit as one of the matrices. In order to do that we first prove the following lemma:

\begin{lemma}[Matrix unit lemma]\label{mul}

Let $A$ be a $d\times d$ matrix with elements from an arbitrary field, $r \in \mathbb{N}$. Then 
\begin{align}
        E_{jk} A^r E_{jk} = (A^r)_{kj}E_{jk},
\end{align}
 where $(A^r)_{kj}$ is the $(k,j)-$th entry of the matrix $A^r.$
\end{lemma}  

\begin{proof}
 \begin{align*}
    E_{jk}A^rE_{jk}= E_{jk}(\sum_{i,l} (A^r)_{il}E_{il})E_{jk}=E_{jk}(\sum_{i} (A^r)_{ij}E_{ik})=(A^r)_{kj}E_{jk}.
\end{align*}   
\end{proof}

Now we are ready to formulate the theorem which provides a TI MPS representation with PBC for the $W$-state with lower bond dimension than ones known to the authors so far.

\begin{theorem}\label{W state representation}
For arbitrary $n \in \mathbb{N}$ we can have the following TI MPS representation with PBC for the non-normalized $W$-state of order $n$ using $(\floor*{\frac{n}{2}}+1) \times (\floor*{\frac{n}{2}}+1)$ matrices:

    $$
    A_{0} =\displaystyle \begin{pmatrix}
1 & 1 & 1 & \dotsc  & $x(n)$\\
1 & 0 & 0 & \ddots  & 0\\
0 & 1 & 0 & \ddots  & 0\\
\vdots  & \ddots  & \ddots  & \ddots  & 0\\
0 & 0 & \ldots & 1 & 0
\end{pmatrix},\quad A_{1} =\displaystyle \begin{pmatrix}
0 & 0 & 0 & \dotsc  & 1\\
0 & 0 & 0 & \ddots  & 0\\
0 & 0 & 0 & \ddots  & 0\\
\vdots  & \ddots  & \ddots  & \ddots  & 0\\
0 & 0 & 0 & \dotsc  & 0
\end{pmatrix},
$$

 where $x(n) \in \mathbb{C}$ is one of the roots of the equation $\text{Tr}~(A_0^n) = 0$. 

 Furthermore, matrices $A_0' = \dfrac{2^{-\frac{n - \floor*{\frac{n}{2}} - 2}{n}}}{\sqrt[2n]{n}} A_0$, $A_1' = \dfrac{2^{-\frac{n - \floor*{\frac{n}{2}} - 2}{n}}}{\sqrt[2n]{n}} A_1$ determine a TI MPS representation with PBC for the normalized $W$-state of order $n$

\end{theorem}

\begin{proof}
Here, $A_0 =(\sum\limits_{k=1}^{d - 1} E_{1,k} + E_{k + 1, k}) + x E_{1d}, A_1 = E_{1d},$ where $ d= \floor*{\frac{n}{2}}+1.$

Without loss of generality, from \ref{base formulation} it follows that we have to verify that the following holds: 

\begin{equation*}
\Tr (A_{i_1} A_{i_2} \ldots A_{i_n}) =  \begin{cases*}
 0\ \ \  \forall (i_1, \ldots, i_n) \in \{0, 1\}^n \text{ such that } \sum\limits_{j=1}^n i_j \ne 1\\
 const \ \ \  \forall (i_1, \ldots, i_n) \in \{0, 1\}^n \text{ such that } \sum\limits_{j=1}^n i_j = 1
\end{cases*}
\end{equation*}

and prove that $const = 2^{n - \floor*{\frac{n}{2}} - 2}$.

We consider the following cases based on the number and positions of $A_1$ in the products:
\begin{enumerate}
    
    \item {(\bf single $A_1$)} $I \in \{0, 1\}^n \text{ such that }  \sum\limits_{j=1}^n i_j = 1.$ 
    
     Due to the cyclicity of the trace, our expression will have the form $$ \Tr (A_0^{n - 1} A_1) $$ 
     
    Further, $$\Tr (A_0^{n - 1} A_1) =\Tr (A_0^{n - 1} E_{1d})= (A_0^{n - 1})_{d1}.$$

 In order to prove the needed result, we need to consider the following two properties of our matrices:

{ \bf Descending Row property:}
For $i\geq2,$ since $(A_0)_{ij}=\delta_{ij+1}$ it follows that:
\begin{align}\label{n1}
 (A_0^{k+1})_{il}=\sum_j(A_0)_{ ij}(A_0^{k})_{j l}=(A_0^{k})_{ i-1l}\ \ \text{for}\ \ \forall k \in \mathbb{N},\ i=\overline{1,k+1}.   
\end{align}


Essentially, the rows of the matrix $A_0$ "descend" upon taking subsequent powers.

{ \bf Column summation property:}
Let $(C_1,C_2,\ldots,C_d)$ be the columns of a matrix $B.$
It can be directly checked that the columns of $BA_0$ are $(C_1+C_2,C_1+C_3,C_1+C_4,\ldots,C_1+C_D,xC_1).$

From  the column summation property it follows that upon taking subsequent powers of $A_0$ the variable $x$ "migrates" to the position (1,1) in $d-1$ steps (multiplications) and from the descending row property it follows that  $x$ migrates to the position $(d,1)$ in another $d-1$ steps. Prior to taking $2d-2$ steps all entries in the $(d,1)$ position are constants. 

     It can be summarised as
     \begin{align}\label{d1}
        ( A_0^{m})_{d1}=  \begin{cases*}
 0, \ m=\overline{1,d-2} \\
 1, \ m=\overline{d-1,d} \\
 2^k, \ m= d+k, k =\overline{1,d-2}
 \end{cases*}  
      \end{align}
      Taking $m=n-1,$ we conclude that the necessary $d$ is such that $n-1 \leq 2d-2$ i.e. $\frac{n+1}{2}\leq d.$ Taking $d = \floor*{\frac{n}{2}}+1$, then $const = \Tr(A_0^{n - 1} A_1) = (A_0^{n - 1})_{d1} = 2^{n - \floor*{\frac{n}{2}} - 2}$.

    \item {(\bf no $A_1$)}  $I = \{0\}^n. $  
    
    From the column summation property it directly follows that $( A_0^{m+1})_{dd}=x(n)( A_0^{m})_{d1}.$
    
    \begin{align}
        ( A_0^{m})_{dd}=  \begin{cases*}
 0, \ m=\overline{1,d-2} \\
 x(n), \ m=\overline{d-1,d} \\
 2^{k}x(n), \ m= d+k, k =\overline{1,d-2}
 \end{cases*}  
    \end{align}
    From the above it follows that if $\frac{n+1}{2}\leq d, \Tr (A_{0} A_{0} \ldots A_{0})$ is either $0$ or contains at least one term $2^lx(n).$ Owing to the column summation property, no negative terms arise in the equation $\Tr (A_{0} A_{0} \ldots A_{0})=0$. Since, we are working over the field of complex numbers, the equation always has a solution. This solution determines the value of $x(n).$
    
        \item {\bf (consecutive $A_1$)}  $I \in \{0,1\}^n \text{ such that } \exists k : i_k=i_{k+1}=1$ or  $i_1+i_n=2.$ 
    
    From the nilpotency of $A_1,$ the trace in this case will trivial be $0.$
    
    \item {\bf (sparse $A_1$)} $I \in \{0,1\}^n \text{ such that } \not \exists k : i_k=i_{k+1}=1, i_1+i_n \neq 2$ and $\sum\limits_{i_k:i_k=1}i_k\geq 2.$ 
    
    Owing to the cyclicity of the trace, the product has the form $A_1 A_0^{r_1}A_1 A_0^{r_2} \ldots A_1 A_0^{r_k}$ where $\sum\limits_{k}r_k \leq n-2.$ Since, $A_1=E_{1d}$ from the matrix unit lemma \ref{mul}, it follows that
    \begin{align}\label{sptr}
       Tr(A_1 A_0^{r_1}A_1 A_0^{r_2} \ldots A_1 A_0^{r_k}) = (A_0^{r_1})_{d1}(A_0^{r_2})_{d1}\ldots (A_0^{r_k})_{d1}  
    \end{align}
   
       As the sum of the powers of the matrices in the expression $Tr(A_1 A_0^{r_1}A_1 A_0^{r_2} \ldots A_1 A_0^{r_k})$ is $n$ and the number of expressions of the form $A_0^{r_j}$ is at least 2, we have $r_{min}=\min\{r_1,r_2,...r_k\}<\floor*{\frac{n}{2}}.$ But from \eqref{d1} it follows that $(A_0^{r_{\min}})_{d1}=0$ if $\floor*{\frac{n}{2}}<d-1.$ Hence,
        $$
    Tr(A_1 A_0^{r_1}A_1 A_0^{r_2} \ldots A_1 A_0^{r_k}) = (A_0^{r_1})_{d1}(A_0^{r_2})_{d1}\ldots (A_0^{r_k})_{d1} =0.
    $$
\end{enumerate}
\end{proof}

In this way, we obtain an analogous result in asymptotics $d(n) = O(n)$ for TI MPS representation with PBC for the $W$-state but with a better constant factor $\dfrac{1}{2}$. In Section 4, we discuss the results of numerical experiments with small $n$ for determining the minimum possible bond dimension $d(\psi)$ for the $W$-state. These experiments lead us to conjecture that the representation obtained in Theorem \ref{W state representation} could be optimal both in terms of asymptotics and the constant factor or in other words it is impossible to come up with a TI MPS representation with PBC with bond dimension smaller than $\floor*{\frac{n}{2}}+1$ for the $W$-state of order $n$.

We will generalize our approach, which was used in the last theorem. As we see, sufficiently good estimates were obtained by using a sparse nilpotent matrix, i.e. a matrix unit as one of the matrices. Let us call \textit{TI MPS representation with PBC with a scaled matrix unit} such a TI MPS representation  with PBC in which one of the matrices is proportional to some matrix unit $E_{j, k},\ j \ne k$. In the case where $i = j$ the matrix will not be nilpotent, which makes this case potentially less suitable for generating interesting solutions. The formal use of a matrix proportional to the matrix unit is natural because quantum states are often considered up to normalization. It turns out that in such a situation it is possible to formulate a structural theorem that makes it easier to find TI MPS representation with PBC or at least to check whether the given matrices can be used to build such a representation.

\begin{theorem}[Necessary and sufficient condition for TI MPS with PBC representation with a scaled matrix unit]\label{base i-sparse TI MPS construction}

Let $A_1 = \lambda E_{jk}$ and $ \lambda \in \mathbb{C} / \{0\},\ j \ne k.$  Then the matrix $A_{0}$ will form a TI MPS representation with PBC along with $A_1$ for a TI state $\psi \in (\mathbb{C}^2)^{\otimes n}$ if and only if the matrix $A_{0}$ and the state $\psi$ satisfy the following two conditions :

1) $c_I=0$, if in the tuple $I \in \{0,1\}^n$ two consecutive $1$' s (upto cyclicity) appear.

2) the following system of equations is satisfied:
\begin{equation}\label{2eq}
\begin{cases*}
\Tr A_{0}^n = c_{I} \ \text{for} \  I = (0,0,\ldots 0) \\
\prod\limits_{m=1}^l (A_{0}^{p_m})_{k,j} = \lambda^l c_I\ \text{for all }I  = (1, \underbrace{0, \ldots, 0}_{p_1}, 1, \underbrace{0, \ldots, 0}_{p_2}, \ldots, 1, \underbrace{0, \ldots, 0}_{p_l}),
\end{cases*}
\end{equation}

\begin{equation*}
\text{where} \ p_m > 0,\ l > 0 \ \text{and} \ \sum p_m = n - l.
\end{equation*}

\end{theorem}

\begin{proof}

As usual, for $A_0$ and $A_1$ to be a TI MPS representation with PBC of the state state $\psi \in (\mathbb{C}^2)^{\otimes n}$ it is necessary and sufficient for each of the following equations to hold:
\begin{align}\label{base2}
\hspace*{-1cm} \Tr (A_{i_1} A_{i_2} \ldots A_{i_n}) = c_{(i_1, \ldots, i_n)} \ \ \ \forall (i_1, \ldots, i_n) \in \{0, 1\}^n.
\end{align}
Let us divide the system of equations \eqref{base2} into 3 classes of equations and let us prove that the conditions \eqref{2eq} are necessary and sufficient for \eqref{base2} to hold using the approach that was used in the theorem above.

\begin{enumerate}
    \item {(\bf no $A_1$)}  $I = \{0\}^n. $ 

    The first point in condition 2) coincides with holding of the corresponding equation in \ref{base2}.

     \item {\bf (consecutive $A_1$)}  $I \in \{0,1\}^n \text{ such that } \exists k : i_k=i_{k+1}=1$ or  $i_1+i_n=2$. 

    In analogy to Theorem 1, condition 1) is necessary and sufficient for the corresponding equation due to the nilpotency of $A_1$.

    \item {\bf (sparse $A_1$)} $I \in \{0,1\}^n \text{ such that } \not \exists k : i_k=i_{k+1}=1, i_1+i_n \neq 2.$

   In a manner similar to the deduction in Theorem \ref{W state representation} and using \ref{mul}, we arrive at \begin{equation}\label{sptr_i}
       Tr(A_1 A_{0 }^{p_1}A_1 A_{0}^{p_2} \ldots A_1 A_{0}^{p_l}) = \lambda^l (A_{0}^{p_1})_{kj}(A_{0}^{p_2})_{kj}\ldots (A_{0}^{p_l})_{kj}  
    \end{equation}

    It follows that the second point in condition 2) is equivalent to this class of equations being true.
\end{enumerate}

\end{proof}

This theorem opens up a simpler and more efficient way for us to construct TI MPS representations for a large class of states, which has a number of advantages over the general methods for constructing TI MPS representations.

\begin{remark}
    In the formulation of Theorem \ref{base i-sparse TI MPS construction} we can interchange $0$ and $1$ (i.e. $ A_0 \leftrightarrow A_1$  and $I=(i_1,i_2, \ldots, i_n) \leftrightarrow I'= (1-i_1,1-i_2, \ldots, 1-i_n))$. The proof of the theorem will remain the same. All results below will also hold upon such interchange.
\end{remark}

If the state $\psi$ satisfies the system of equations \ref{2eq} along with matrices $A_0$ and $A_{1}$ then we can denote $\gamma_i=(A_{0}^{i - 1})_{k,j}$ and all equations except one become equations in the $n-1$ variables $\gamma_1,\ldots, \gamma_{n-1}$. This naturally leads to the question of whether we can use lower-dimensional matrices instead of the ones guaranteed by Theorem \ref{base i-sparse TI MPS construction}. It turns out that for an arbitrary state admitting TI MPS representation with PBC with a scaled matrix unit, one can always construct an $n$-dimensional representation.



\begin{theorem}[Canonical construction of the TI MPS representation with PBC with a scaled matrix unit ]\label{canonical construction with matrix unit}

 Let for $\psi \in (\mathbb{C}^2)^{\otimes n}$ there exist a TI MPS representation with PBC with a scaled matrix unit with matrices  $B_1 = \lambda E_{j, k}, B_{0}$ where $ \lambda \in \mathbb{C} / \{0\}$. Then for $\psi$ there exists a TI MPS representation with PBC with a scaled matrix unit using matrices of dimension $n \times n$ of the following form:

$$
A_{0} = \lambda \displaystyle \begin{pmatrix}
0 & \gamma_{n - 1} & \gamma_{n - 2} & \gamma_{n - 3} & \dotsc  & \gamma_1\\
0 & \omega & 1 & 0 & \ddots  & 0\\
0 & 0 & 0 & 1 & \ddots  & 0\\
\vdots  & \ddots  & \ddots  & \ddots  & \ddots & \ddots\\

\vdots  & \ddots  & \ddots  & \ddots  & 0 & 1\\
0 & 0 & \ldots & 0 & 0 & 0
\end{pmatrix},\quad A_{1} = \lambda \displaystyle \begin{pmatrix}
0 & 0 & 0 & 0 & \dotsc  & 0\\
0 & 0 & 0 & 0 & \ddots  & 0\\
0 & 0 & 0 & 0 & \ddots  & 0\\
\vdots  & \ddots  & \ddots  & \ddots  & \ddots & \ddots\\

\vdots  & \ddots  & \ddots  & \ddots  & 0 & 0\\
1 & 0 & \ldots & 0 & 0 & 0
\end{pmatrix}
$$
where $\gamma_1, \ldots, \gamma_{n - 1}, \omega \in \mathbb{C}$ satisfy $\Tr A_{0}^n = \Tr B_{0}^n = (\lambda\omega)^n$, $(A_{0}^q)_{1, n} = (B_{0}^q)_{k, j} = \lambda^q \gamma_q$ for all $q = \overline{1,n - 1}$.
\end{theorem}

\begin{proof}

From the assumptions of the theorem, we have a matrix pair $B_1 = \lambda E_{jk}$ and $B_{0}$, which together with $\psi$ satisfy the conditions of Theorem \ref{base i-sparse TI MPS construction}. We can consider according to Remark \ref{1} a similar TI MPS representation with PBC with a scaled matrix unit for $\dfrac{1}{(\lambda)^n} \psi$ using the matrices $C_1 = \dfrac{1}{\lambda} B_1 = E_{j,k}$ and $C_0 = \dfrac{1}{\lambda} B_{0}$. It suffices for us to prove that the matrices $\dfrac{1}{\lambda}A_1 = E_{n,1}$ and $\dfrac{1}{\lambda} A_{0}$ for some $\gamma_1, \ldots, \gamma_{n - 1}$ and $\omega$ determine a TI MPS representation with PBC for $\dfrac{1}{(\lambda)^n} \psi$. To prove this, we need to check that Theorem \ref{base i-sparse TI MPS construction} holds for $D_0 = \dfrac{1}{\lambda} A_{0}, D_1 = \dfrac{1}{\lambda}A_1$ and $\psi' = \dfrac{1}{(\lambda)^n} \psi $. Additional properties of the relation between the matrices $A_{0}$ and $B_{0}$ will be proved in a related way in the process.

We can rewrite $D_0$ in the language of matrix units $$D_{0} = \sum\limits_{j=1}^{n -1} \gamma_{n - j} E_{1, j + 1} + \sum\limits_{j=2}^{n -1} E_{j, j + 1} + \omega E_{2,2}.$$ We can check the following conditions of Theorem \ref{base i-sparse TI MPS construction} for $D_0, D_1$ and $\psi'.$

\begin{enumerate}
    \item ($c_I^{\psi'}=0$, if in $I$ two consecutive $i$' s (upto cyclicity) appear).

     From the assumptions of the theorem, $ B_0$ and $B_1$ determine a TI MPS  representation with PBC with a scaled matrix unit for $\psi$.  Therefore, owing to  condition 1) from Theorem \ref{base i-sparse TI MPS construction} which states that $c_I^{\psi} = 0$ it follows that  {$c_I^{\psi'} = \dfrac{1}{(\lambda)^n} c_I^{\psi} = 0.$ }

    \item ($\Tr D_{0}^n = c_{(0, 0, \ldots, 0)}^{\psi'}$)

    All the terms in the matrix unit decomposition of $D_0$ other than $\omega E_{2,2}$  are scaled matrix units $E_{k,j}$ with $ k < j$. Therefore, due to $E_{rs}E_{sl} =E_{rl}$, it follows that $\Tr (D_{0}^n) = \omega^n $ since other matrix units in the decomposition cannot contribute to the diagonal. 
    Therefore, to fulfill the condition, we only need to put $w = \sqrt[n]{c^{\psi'}(0, \ldots, 0)}$. 
    
    Additionally, the condition from the theorem $\Tr A_{0}^n = \Tr B_{0}^n = (\lambda \omega)^n$ can easily be obtained by expressing $A_{0}$ as $\lambda D_0$ and $B_{0}$ as $\lambda C_{0}$.

    \item ($\prod\limits_{m=1}^l (D_{0}^{p_m})_{1,n} = c_I^{\psi'} $)

     We already have a pair of matrices $C_1$ and $C_{0}$, which together with $\psi'$ satisfies the conditions of Theorem \ref{base i-sparse TI MPS construction} (with the scaling factor of the matrix unit $C_1$ being $1$) and hence $\prod\limits_{m=1}^l (C_{0}^{p_m})_{k, j} = c^{\psi'}_I\ \text{for all }I = (1, \underbrace{0, \ldots, 0}_{p_1}, 1, \underbrace{0, \ldots, 0}_{p_2}, \ldots, 1, \underbrace{0, \ldots, 0}_{p_l})$ where $p_m > 0, l > 0$ and $\sum p_m = n - l$.

    Assume $\gamma_q = (C_{0}^{q})_{k,j}$, then to prove this condition for $D_1$ and $D_{0}$ it is sufficient to prove that $(D_{0}^{q})_{1,n} = \gamma_q$ when $q = \overline{1,n - 1}$. 

    Note that \begin{multline}\label{DD}
     D_0=\gamma_{n-(2-1)}E_{1,2}+\gamma_{n-(3-1)}E_{1,3}+\ldots  \gamma_2 E_{1,n-1}+\\ + \gamma_{n-(n-1)} E_{1,n} +E_{2,3}+E_{3,4}+\ldots + E_{n-1,n}+
     \omega E_{22}
    \end{multline}
    It is evident that 
    \begin{align}\label{D1n}
        (D_{0}^{q})_{1,n}
        &= (D_0(\sum\limits_{j=2}^{n -1} E_{j, j + 1} + \omega E_{2,2})^{q-1})_{1,n}  \nonumber \\
        &= (\gamma_{n-(n-(q-1)-1)}E_{1,n-(q-1)} E_{n-(q-1),n})_{1,n} \\
        &= \gamma_q. \nonumber 
    \end{align}
The first equality follows from the rule of matrix unit multiplication and the fact that in \eqref{DD} for all $E_{i,j}$ we have $j>1.$ The second equality holds since we only have to pick $E_{n-1,n}$ from the last bracket as we are concerned with the element $(D_{0}^{q})_{1,n},$ each multiplication starting from the right reduces the row number of the matrix unit by $1.$ For $\omega$ (the coefficient of $E_{22}$) to enter the expression \eqref{D1n}, $2 \geq n-(q-2)$ or $q \geq n$ which does not hold as we have $q = \overline{1,n - 1}.$ After $q-1$ multiplications we are left with the term $E_{n-(q-1),n}.$ 
Thus, $(A_{0}^q)_{1, n} = ((\lambda D_{0})^{q})_{1, n} = \lambda^q \gamma_q.$

\end{enumerate}

\end{proof}

Thus, for all states for which it is possible to construct a TI MPS representation with PBC with scaled matrix unit, it is possible to do so using matrices of size $n \times n$. This gives rise to additional possibilities for the construction or verification of TI MPS with PBC. Let us consider an example to illustrate this.

\begin{example}\label{Ushift}
A scaled upper-shift matrix $A_0 = \dfrac{1}{\sqrt[2n]{n}} \sum\limits_{k=1}^{n - 1} E_{k, k + 1}$ and scaled matrix unit $ A_1 = \dfrac{1}{\sqrt[2n]{n}}E_{n,1}$ determine a TI MPS representation with PBC for $W-$state of order $n$.
\end{example}

Below we provide sketches of several proofs of this fact to illustrate how the developed apparatus can be used:

1) We can use the definition of TI MPS with PBC and leverage the upper-shift matrix property of "raising the rows" of the matrix unit. This is not a very complicated proof, but would require some minimal verification nonetheless.

2) We can use Theorem \ref{base i-sparse TI MPS construction}. It boils down to considering fairly obvious properties of the powers of the upper-shift matrix making the proof almost trivial.

3) Finally, we can use the fact that this example is exactly a special case of the construction in Theorem \ref{canonical construction with matrix unit} with coefficients that are directly obtained from the representation from Theorem \ref{W state representation}. Theorem \ref{canonical construction with matrix unit} provides a less efficient representation than Theorem \ref{W state representation}, which however does not prevent one from using this approach here in purely formal terms and getting the results immediately.

\begin{remark}\label{sense}
 Theorem \ref{canonical construction with matrix unit} guarantees that a representation can be constructed with matrices of dimension $n$ (which implies that $d(\psi) \le n$)  for those states $\psi$ that admit TI MPS representation with PBC with a scaled matrix unit. Thus, it allows to construct a more efficient representation if any representation with a scaled matrix unit has been obtained. Finally, we can look for a representation exactly in the canonical construction using Theorem \ref{canonical construction with matrix unit}. This potentially allows us to search for a solution in a class of matrices of a fairly simple form with a reduced number of degrees of freedom. This especially makes sense if we know that a representation with a matrix unit exists for a given state.

\end{remark}

Let us put together the differences between the usual TI MPS with PBC search and the search for representations using matrix units.

\begin{enumerate}
    \item The resulting system of equations is significantly simpler. Due to the nilpotency of $A_i$, a large number of equations automatically become true for TI states for which 1) holds in Theorem \ref{base i-sparse TI MPS construction}. The remaining equations include only a single element of the matrix $A_{0}$ raised to various powers, the maximum of which is $n.$ In other words, everything depends on $(A_{0})_{k,j}, (A_{0}^2)_{k,j}, \ldots, (A_{0}^{n - 1})_{k,j}$ (except for an equation for the trace of the $n$-th power of the matrix $A_{0}$). Thus, we get rid of the need to consider the traces of different products of these matrices and pass to a simpler system of equations.
    \item All equations in the system (except one) when using the Theorem \ref{base i-sparse TI MPS construction} depend on the number of all possible $I  = (1, \underbrace{0, \ldots, 0}_{p_1}, 1, \underbrace{0, \ldots, 0}_{p_2}, \ldots, 1, \underbrace{0, \ldots, 0}_{p_l})$. After discarding identical and incompatible ones and using the fact that any permutation of $(p_1,p_2,\ldots, p_l)$ does not change the LHS of the equation:
    \begin{align}
        \prod\limits_{m=1}^l (A_{0}^{p_m})_{k,j} = \lambda^l c_I,\
    \end{align}
      
   we can estimate the number of equations from above as the number of partitions of $n = \sum_{m=1}^l p_l + l$. This number is asymptotically bounded from above as $O(\frac{e^{c_0\sqrt{n}}}{n})$ where $c_0 = \pi \sqrt{\frac{2}{3}}$ by the Hardy-Ramanujan Asymptotic Partition Formula \cite{HardyAsymptoticFI}. This is significantly lower than the number of different equations in \ref{base TI formulation}, which are $O(\frac{1}{n}\sum _{p|n}\varphi (p)2^{n/p})$ by \ref{numb}. This becomes clearer when comparing the asymptotics of the logarithms of the number of equations; in the case of the classical formulation, it grows as $O(n)$, while in the case of Theorem \ref{base i-sparse TI MPS construction} it is bounded from above by $O(\sqrt{n})$.
    Furthermore, we can reduce the dimensionality of matrices in a representation to $n$, if the current representation has a bond dimension greater than $n$, using Theorem \ref{canonical construction with matrix unit}. In fact, those states for which such an approach is applicable are guaranteed to allow the construction of representations with a small number of variables and smaller matrix dimensions compared with those guaranteed in the general case, which improves the quality of these representations and simplifies the search.
    \item In itself, the use of matrix units is convenient in that the matrix has only one non-zero element, which potentially facilitates its storage in memory. With another matrix that is sufficiently sparse to pair with it, we can get MPS, which can be stored efficiently in terms of memory. For any state allowing such a representation, there exists a representation whose matrices contain $O(n)$ non-zero elements, according to Theorem \ref{canonical construction with matrix unit}, which is way lesser than the potential maximum $O(d^2)$.
\end{enumerate}

\begin{table}[h!]
\begin{center}
\footnotesize
\begin{tabular}{ |p{5cm}|p{5cm}|p{5cm}| } 
\hline
&TI MPS with PBC in general form
by definition & TI MPS with PBC
with scaled matrix unit by Theorem \ref{base i-sparse TI MPS construction}\\
\hline
Searching for a representation & Solving a large system of equations on the traces of the product of matrices & Solving a system of elementary polynomial equations for a certain position in the matrix in different degrees\\
\hline
For which states can be constructed & All & A certain class of ''sparse'' states in terms of basis vectors with nonzero coefficients\\
\hline
The dimensionality of the best known general construction for all states that admit such a representation& $O(n 2^{\frac{n}{2}})$ & $O(n)$\\
\hline
Logarithm of the number of equations to find TI MPS & $O(n)$ & $O(\sqrt{n})$\\
\hline
Number of representation matrix elements to store & $2 d^2$ & $d^2 + 1$, by canonical construction can be reduced to $O(n)$\\
\hline

\hline
\end{tabular}
\end{center}
\caption{Table comparing the straightforward construction of TI MPS with PBC in the general case and TI MPS with PBC with scaled matrix unit}
\end{table}

The set of TI states that admit a TI MPS representation with PBC with a scaled matrix unit as one of the matrices forms a subset of all TI states. This subset is not empty (e.g., by Theorem \ref{W state representation}). But not all TI states have such a representation. For example, we can consider $\psi$ with $c_I^{\psi} = 1 \ \forall I$. It can be observed that $\psi$ is a $TI$ state, but at the same time does not satisfy condition 1) in Theorem \ref{base i-sparse TI MPS construction}, and therefore does not admit a TI MPS representation with PBC with a scaled matrix unit.

Those TI states $\psi$, which allow TI MPS representation with scaled matrix unit, must satisfy condition 1) in Theorem \ref{base i-sparse TI MPS construction}, that is, those basis vectors $\ket{i_1i_2\ldots i_n}$ with non-zero coefficients must be "sparse". In this case the TI state itself does not have to be highly sparse in the usual sense of the word, although according to the above, a large count of coordinates are nullified. From the point of view of certain heuristic considerations (in particular, the natural possibility to nullify many coordinates) this way of obtaining TI MPS representations should suit well for those states which are sparse in the usual sense of the word (for example, $W$-state for which such a construction works well since it has only $n$ nonzero coordinates from $2^n$ possible).


This leads to the following open problem:

\begin{problem}
Does there exist an explicit description of the set of all TI states $ \psi \in (\mathbb{C}^2)^{\otimes n} $, which have TI MPS representation with PBC  in which one of the matrices is a scaled matrix unit $\lambda E_{i, j},\ i \ne j$ (which is equivalent to satisfying the conditions of Theorem \ref{base i-sparse TI MPS construction})?
\end{problem}

 We can describe this set as precisely those TI states $ \psi \in (\mathbb{C}^2)^{\otimes n} $ for which matrix $A$ is used to satisfy the conditions of Theorem \ref{base i-sparse TI MPS construction}. However, this description is not free from the search for a certain matrix $A$, so a simpler description in an explicit form of this set is preferable.

\section{Algorithmic finding of the optimal bond dimension in the general case}

The methods described above for finding a TI MPS representation with PBC for the $W$-state or more generally for a certain class of TI states do not work for all states. Furthermore, it is not guaranteed to be optimal for $d(\psi)$ and, as far as is known to the authors, methods to construct optimal TI MPS representations with PBC deterministically do not exist at the moment. We build a simple method to construct, by some asymptotics, the optimal $d(\psi)$.

To this end, we will need Hilbert's nullstellensatz \cite{Hilbert1893, Vinberg}. We provide its weak form below:

 {\bf Weak Nullstellensatz}. 
 
 {\it The ideal $ I\subset k[X_{1},\ldots ,X_{n}]$ contains $1$ if and only if the polynomials in $I$ do not have any common zeros in $K^n$, where $K$ is an algebraically closed field extension of $k$.}

 This theorem establishes that a system of polynomial equations has a solution if and only if the basis of the ideal generated by it contains $1$, which makes it possible to deterministically check the existence of solutions for systems of polynomial equations using the construction of Gröbner bases\cite{Groebner}.

 Since the system of equations \ref{base TI formulation} that the matrices $A_0$ and $A_1$ must satisfy in order to be TI-MPS with PBC for a given state $\psi$ is a polynomial system of equations in the elements of these matrices, we can formulate the following theorem.
 
 \begin{theorem}\label{bb}
 
For a TI state $\psi \in (\mathbb{C}^2)^{\otimes n}$ one can deterministically find $d(\psi)$ with time complexity $O(n^{2^{(Mn^2 2^ n)(1 + o(1)))}})$, where $M$ is some constant, according to the following algorithm:

Starting from $d=1$, we iteratively search for the Gröbner basis for the system \ref{base TI formulation} with respect to the matrices $A_0$ and $A_1$ with arbitrary elements, that is, with $2 d^2$ unknowns.

If this basis contains $1$, we move to $d + 1$.

In the case of absence of $1$ in the basis, $d(\psi) = d$.

\end{theorem}

\begin{proof}
    We know that $d(\psi) \in O(n2^{\frac{n}{2}})$ (which follows from Theorem 1 and Theorem 3 from \cite{Gar+06}).

According to the algorithm described in the theorem, we certainly reach $d$ for which there is a TI MPS with PBC representation with matrices of this dimension, which would require no more than $O(n2^{\frac{n}{2}})$ updates of $d$ (since $d(\psi) \in O(n2^{\frac{n}{2}})$). It follows from Weak Nullstellensatz that the algorithm is correct, that is, that the first time the Gröbner basis does not contain $1$, we will have non-trivial solutions, that is, matrices $A_0$ and $A_1$ such that they will determine a TI MPS representation with PBC. The transition of the algorithm from the minimum possible $d$ to the next minimal unchecked $d + 1$ guarantees that the first $d$ found will be optimal.

Let us estimate the asymptotics of the algorithm. The complexity of the Buchberger algorithm to find the Gröbner basis can be estimated as $O(k^{2^{m + o(m)}})$ with a maximum degree of polynomials $k$ and the number of variables $m$\cite{GroebnerBound}. In our case, at all steps $k=n$ (which follows from the fact that in the product of $n$ matrices all elements are polynomials of degree $n$  in matrix elements), and $m=2d^2$. Thus, the time complexity of each step is $O(n^{2^{2d^2 + o(d)}})$. In the worst case, the complexity of all steps is $O(\sum\limits_{d=1}^{C n2^{\frac{n}{2}}} n^{2^{2d^2 + o(d)} })$, where $C$ is a constant, which can be roughly estimated from above as $O(C n2^{\frac{n}{2}} n^{2^{2(C n2^{\frac{n} {2}})^2 + o(d)}})$, which can be estimated as $O(n^{2^{(Mn^2 2^n)(1 + o(1)))}})$ , where $M$ is a constant.
\end{proof}

As we can see, the estimates of the asymptotics of the time to find $d(\psi)$ are very large, which tells us both about the potential need to improve this estimate and the possible search for more efficient algorithms. In this regard, the following question can be formulated.

\begin{problem}
Are there more efficient algorithms in terms of the complexity of obtaining $d(\psi)$ than the one presented in Theorem \ref{bb}?
\end{problem}

A potential improvement in the algorithm running time estimate from Theorem \ref{bb}  can occur both based on the fact that the initial estimate was obtained from very general estimates for intermediate parts of the algorithm, and based on the special structure of the system of equations obtained from conditions which allow for TI MPS representations with PBC. This may also allow us to potentially modify the original algorithm. Nevertheless, it is also possible to use approaches that do not use the Nullstellensatz.

Despite the fact that we have built an algorithm for obtaining $d(\psi)$, we still do not know the algorithm for finding the optimal TI MPS with PBC, that is, TI MPS with PBC, in which the matrix dimensions are equal to $d(\psi)$ .

\begin{problem}
Is there an algorithm for constructing a TI MPS representation with PBC, in which the dimensions of the matrices are optimal (i.e. equal to $d(\psi)$)  for a given TI state $\psi \in (\mathbb{C}^2)^{\otimes n}$ ?
\end{problem}

Despite the fact that we have built an algorithm for finding $d(\psi)$ for an arbitrary TI state, there are currently no known results on the asymptotic behavior of $d(n)$ for a $W$ state. The current results\cite{CiracMPSReview} are such that $d(n)$ can be estimated below as $\Omega(n^\frac{1}{3 + \delta})$ for any $\delta > 0$, which was proved in \cite{ShitovMPS}. At the same time, the conjecture that $d(n) \in \Omega(n^\frac{1}{3})$\cite{Gar+06} remains unproved. The current best estimate from above known to the authors is the result of Theorem \ref{W state representation}, that is, $d(n) \in O(n)$ with constant factor $\dfrac{1}{2}$.

The approach from Theorem \ref{bb} allows us to search for $d(\psi)$ for any state. We present a table with the results of computing this number for the $W$-state for small $n$.

We used Theorem \ref{bb} along with the fact that there already exists a representation for the W state with bond dimension $\floor*{\frac{n}{2}}+1$. In this way, we could have started with $d=1$ and continue till the bond dimension in our computation. We carried out the Groebner basis computation using \textit{Wolfram Mathematica}\cite{wolfram}.

It turns out that for all cases except when $n=6$ the obtained values of $d(n)$ coincides with dimension of the representation obtained in Theorem \ref{W state representation}. The case where $n=6$ the dimension of the optimal representation is either three or four. This is related with computational difficulties in evaluation of Groebner basis for $n=6$ and $d=3$. The results obtained hint that the representation obtained in Theorem \ref{W state representation} can be optimal asymptotically and in terms of the constant factor. Moreover, it may be optimal even in terms of reaching the lower bound $d(n).$

\begin{table}[h!]
\begin{center}
\begin{tabular}{ |c|c| } 
\hline
$n$ & $d(n)$ \\
\hline
2 & 2 \\
3 & 2 \\
4 & 3 \\
5 & 3 \\
6 & 3 - 4 \\
7 & 4 \\
8 & 5 \\
9 & 5 \\
\hline
\end{tabular}
\end{center}
\caption{Table with $d(n)$ for $W$-state for small $n$.}
\end{table}

The above facts lead to the following problems.

\begin{problem}
Is the estimate obtained in Theorem \ref{W state representation} optimal? Or, in other words, is it true that $d(n)=\floor*{\frac{n}{2}}+1$ for the $W$-state? In particular, is it true that $d(6)=4?$ 
\end{problem}

\begin{problem}
Are the estimates obtained in Theorem \ref{W state representation} asymptotically optimal? Or, in other words, is it possible to obtain better estimates for $d(n)$ than $O(n)$ with a constant factor $\dfrac{1}{2}$?
\end{problem}

Also, when optimizing calculations, $d(n)$ for the $W$-state for a number of small values of $n$ can be calculated, which can be useful in practice.

It can be noted that when calculating specific $d(n)$ for the $W$-state, we see that $d(n) \le d(n + 1).$ Wich leads us to the following problem.

\begin{problem}
Is $d(n)$ for a $W$-state monotonically non-decreasing?
\end{problem}

A similar consideration may help to construct more efficient ways to calculate $d(n)$.

\section{Conclusion}

In this work, we explore the theory of Matrix Product States (MPS), with a particular focus on the construction of translationally invariant MPS representations with periodic boundary conditions (PBC) for given states. We primarily focus on determining the optimal bond dimension $d(\psi)$ for a given MPS representation.

Starting with an analysis of the $W$ state, we develop a representation using the MPS formalism and discuss its implications, particularly with respect to the relationship between the bond dimension and the size of the state. This has repercussions for computational efficiency.

Subsequently, we generalize our approach to examine arbitrary states. By appealing to the Weak Nullstellensatz and the concept of Gröbner bases, we develop a deterministic algorithm to find the optimal bond dimension for arbitrary states. Although the time complexity of this algorithm, as it stands, is rather prohibitive, it provides a starting point for future refinements and enhancements.

Despite the successful development of the algorithm, we have also highlighted that there is currently no known method for constructing an optimal TI MPS representation with PBC, a promising area for future research. In parallel, exploring the possibility of refining the complexity estimate of the algorithm, which remains an open question, is equally interesting.

Our investigations further hint that the representation developed for the $W$ state may be optimal, at least for the range of small values of $n$ we have explored. The potential optimality of the representation and its implications are yet to be fully realized and call for further exploration.

Our analysis of specific $d(n)$ for the $W$ state suggests that $d(n)$ is a non-decreasing function, a property that, if generalizable, could offer additional tools for constructing more efficient algorithms to calculate $d(n)$.

The research presented in this work opens up intriguing possibilities for future research, particularly in algorithmic optimization, quantum state representation, and the use of algebraic methods in quantum information theory. While several important questions remain open, we believe that the concepts, methods, and findings discussed here provide a stepping-stone towards a deeper understanding and more effective utilization of Matrix Product States in quantum computation and quantum information theory.

\section*{References}
\bibliography{refs.bib}

\bibliographystyle{unsrt}
\end{document}